\documentclass[aps,prd,eqsecnum,notitlepage,nofootinbib,articles]{revtex4}
\usepackage[centertags]{amsmath}
\usepackage{amssymb}
\usepackage{amsfonts,mathrsfs}
\usepackage{hyperref}	
\usepackage{graphicx}
\usepackage{dcolumn}
\usepackage{bm}
\usepackage{tabularx}
\usepackage{array}
\usepackage{float}
\usepackage{braket}
\usepackage{color}
\usepackage{enumerate}
\newcommand{\blu}{\color{blue}}

\newcommand{\be}{\begin{equation}}
\newcommand{\ee}{\end{equation}}
\newcommand{\bea}{\begin{eqnarray}}
\newcommand{\eea}{\end{eqnarray}}
\newcommand{\bes}{\begin{subequations}}
\newcommand{\ese}{\end{subequations}}
\usepackage{enumitem}
\usepackage{amsthm}

\newtheorem{prop}{Proposition}[section]

\newlist{inlineroman}{enumerate*}{1}
\setlist[inlineroman]{itemjoin*={{, and }},afterlabel=~,label=\roman*.}

\newlist{Inlineroman}{enumerate*}{1}
\setlist[Inlineroman]{itemjoin*={{, and }},afterlabel=~,label=\Roman*.}

\begin{document}
\title{The Volume of Two-Qubit States by Information Geometry}
\author{Milajiguli Rexiti}
\affiliation{School of Advanced Studies, University of Camerino, 62032 Camerino, Italy}
\author{Domenico Felice}
\affiliation{Max Planck Institute for Mathematics in the Sciences, {Inselstrasse 22-04103} Leipzig, Germany;}
\author{Stefano Mancini}
\affiliation{School of Sience and Technology, University of Camerino, 62032 Camerino, Italy\\
INFN-Sezione di Perugia, I-06123 Perugia, Italy}

\begin{abstract}{Using the information geometry approach, we determine the volume of the set of two-qubit states with maximally disordered subsystems. Particular attention is devoted to the behavior of the volume of sub-manifolds of separable and entangled states with fixed purity. 
We show that the usage of the classical Fisher metric on phase space probability representation of quantum states gives the same qualitative results with respect to different versions of the quantum Fisher metric. 
}
\end{abstract}
%\keyword{information theory; Riemannian geometry; entanglement characterization}
%\PACS{{89.70.+c; 02.40.Ky; 03.67.Mn}}

\maketitle

\section{Introduction}

The volume of sets of quantum states is an issue of the utmost importance. It can help in finding separable states within all quantum states. The former are states of a composite system that can be written as convex combinations of subsystem states, in contrast to entangled states \cite{Horo}. The notion of volume can also be useful when defining ``typical'' properties of a set of quantum states. In fact, in such a case, one uses the random generation of states according to a measure stemming from their volume \cite{Lupo}.

Clearly, the seminal thing to do to determine the volume of a set of quantum states is to consider a metric on it. For pure quantum states, it is natural to consider the Fubini--Study metric which turns out to be proportional to the classical Fisher--Rao metric \cite{Marmo1}. However, the situation becomes ambiguous for quantum mixed states where there is no single metric \cite{Petz}. 
There, several measures have been investigated, each of them arising from different motivations \cite{GQ}. For instance, one employed the Positive Partial Transpose criterion \cite{Horo} to determine an upper bound for the volume of separable quantum states and to figure out that it is non-zero \cite{Zycz}. 
{However, this criterion becomes less and less precise as the dimension increases \cite{Ye09}. 
An upper bound for the volume of separable states is also provided in \cite{Szarek} by
combining techniques of geometry and random
matrix theory. Here, separable states are approximated by an ellipsoid with respect to a 
Euclidean product.
Following up, in Reference \cite{Aubrun}, a measure relying on the convex and Euclidean structure of states has been proposed. Such a volume measure turned out to be very effective in detecting separability for large systems. Nevertheless, for two-qubit states, the Euclidean structure is not given by 
the Hilbert--Schmidt scalar product. 
Instead, the latter is considered as a very natural measure upon the finite dimensional Hilbert spaces. For this reason, it has been employed as a natural measure in the space of density operators \cite{Zycz}. Such a measure revealed that the set of separable states has a non-zero volume and, in some cases, analytical lower and upper bounds have been found on it. 
Still based on the Hilbert--Schmidt product, a measure has been recently introduced to evaluate the volume of Gaussian states in infinite dimensional quantum systems \cite{LS15}. 
However, it does not have a classical counterpart, as an approach based on the probabilistic structure of quantum states could have.
This approach has been taken for infinite dimensional quantum systems by using information geometry \cite{FQM17}.} This is the application of differential geometric techniques to the study of families of probabilities \cite{Amari}. Thus, it can be employed whenever quantum states are represented as probability distribution functions, instead of density operators. 
{This happens in phase space where Wigner functions representing Gaussian states were used in such a way that each class of states is associated with a statistical model which turns out to be a Riemannian manifold endowed with the well-known Fisher--Rao metric \cite{FQM17}.}

This approach could be extended to the entire set of quantum states considering the Husimi $Q$-function \cite{Husimi} instead of the Wigner function, being the former a true probability distribution function.
Additionally, it could be used in finite dimensional systems by using coherent states of compact groups $SU(d)$.
Here, we take this avenue for two-qubit systems.
Above all, we address the question of whether such an approach gives results similar to other approaches based on the quantum version of the Fisher metric, such as the Helstrom quantum Fisher metric \cite{Helstrom} and the Wigner-Yanase-like quantum Fisher metric \cite{Luo}.
We focus on states with maximally disordered subsystems and analyze the behavior of the volume of sub-manifolds of separable and entangled states with fixed purity. We show that all of the above mentioned approaches give the same qualitative results. 

The layout of the paper is as follows. In Section \ref{sec2}, we recall the structure of a set of two-qubit states.
Then, in Section \ref{sec3}, we present different metrics on such a space inspired by the classical Fisher--Rao metric. Section \ref{sec4} is devoted to the evaluation of the volume of states with maximally disordered subsystems.
Finally, we draw our conclusions in Section \ref{sec5}.

%%%%%%%%%%%%%%%%%%%%%%%%%%%%%%%%%%%%%%%%%%%%%%%%%%%%%%

\section{Structure of a Set of Two-Qubit States}\label{sec2}

Consider two-qubit states (labeled by 1 and 2) with associated Hilbert space
${\cal H}=\mathbb{C}^2\otimes\mathbb{C}^2$. 
The space ${\cal L}({\cal H})$ of linear operators acting on ${\cal H}$ can be supplied with a scalar product $\langle A,B\rangle={\rm Tr}\left( A^\dag B\right)$ to have a Hilbert--Schmidt space.
Then, in such a space, an arbitrary density operator (positive and trace class operator)
can be represented as follows
\begin{equation}
\rho=\frac{1}{4}\left(I\otimes I+{\boldsymbol r}\cdot{\boldsymbol\sigma}\otimes I+I\otimes
{\boldsymbol s}\cdot {\boldsymbol\sigma}+
\sum_{m,n=1}^3t_{mm}\sigma_{m}\otimes\sigma_{n}\right),
\label{T}
\end{equation}
where $I$ is the identity operator on $\mathbb{C}^2$, $\boldsymbol{r,s}\in\mathbb{R}^3$, and 
${\boldsymbol\sigma}:=(\sigma_1,\sigma_2,\sigma_3)$ is the vector of Pauli matrices. Furthermore,
the coefficients $t_{mn}:={\rm Tr}\left( \rho \, \sigma_m \otimes \sigma_n\right)$ form a real matrix denoted by $T$. 

Note that $\boldsymbol{r, s}$ are local parameters as they determine the reduced states
\begin{eqnarray}
\rho_1&:=&{\rm Tr}_{2} \rho= \frac{1}{2} \left(I+ {\boldsymbol r} \cdot  {\boldsymbol\sigma} \right), \\
\rho_2&:=&{\rm Tr}_{1} \rho= \frac{1}{2} \left(I+ {\boldsymbol s} \cdot  {\boldsymbol\sigma} \right),
\end{eqnarray}
while the $T$ matrix is responsible for correlations.

The number of (real) parameters characterizing $\rho$ is 15, but they can be reduced 
with the following argument. Entanglement (in fact, any quantifier of it) is invariant under
local unitary transformations, i.e., $U_1\otimes U_2$ \cite{Horo}.
Then, without loss of generality,
we can restrict our attention to states with a diagonal matrix $T$. 
To show that this class of states is representative, we can use the following fact: 
if a state is subjected to $U_1\otimes U_2$ transformation,
 the parameters $\boldsymbol{r,s}$ and  $T$ transform themselves as 
\begin{eqnarray}
&&\boldsymbol{r}'=O_1\boldsymbol{r},\\
&&\boldsymbol{s}'=O_2\boldsymbol{s},\\
&&T'=O_1 T O_2^\dag { ,}\label{TOTO}
\end{eqnarray}
where $O_i$s correspond to $U_i$s via 
\begin{equation}
U_i\boldsymbol{ \hat n}\cdot\boldsymbol{\sigma}U_i^\dagger
=(O_i\boldsymbol{\hat n})\cdot\boldsymbol{\sigma}{ ,}
\label{UnsU}
\end{equation}
being $\boldsymbol{ \hat n}$ a versor of $\mathbb{R}^3$.
Thus, given an arbitrary state, we can always choose unitaries $U_1, U_2$ such that the 
corresponding rotations will diagonalize its matrix $T$.
 
As we consider the states with diagonal $T$, we can identify $T$
with the vector $\boldsymbol{t}:=(t_{11},t_{22},t_{33})\in \mathbb{R}^3$. 
Then, the following sufficient conditions are known \cite{Horodecki}:
\begin{enumerate}
\item[i)]
\emph{{For any $\rho$, the matrix $T$ belongs 
to the tetrahedron $\cal T$ with vertices $(-1,-1,-1)$,  
$(-1,1,1)$, $(1,-1,1)$, $(1,1,-1)$.}}
\item[ii)]
\emph{{For any separable state $\rho$, the matrix $T$ belongs to the octahedron 
$\cal O$ with vertices $(0,0,\pm1)$, $(0,\pm1,0)$,
$(0,0,\pm1)$.}}
\end{enumerate}

Still, the number of (real) parameters \eqref{rhotmatrix} characterizing $\rho$ is quite large.
Thus, to simplify the treatment, from now on, we focus on the states with maximally disordered subsystems, namely the states 
with $\boldsymbol{r,s}=0$.
They are solely characterized by the $T$ matrix.
For these states, the following necessary and sufficient conditions are known \cite{Horodecki}:
\begin{enumerate}
\item[i)]
\emph{
Any operator \eqref{T} with $\boldsymbol{r,s}=0$ and diagonal
{$T$ is a state (density operator) iff $T$ belongs to the tetrahedron $\cal T$.}}
\item[ii)]
\emph{{
Any state $\rho$ with maximally disordered subsystems and diagonal
$T$ is separable iff $T$ belongs to the octahedron $\cal O$.}}
\end{enumerate}

%%%%%%%%%%%%%%%%%%%%%%%%%%%%%%%%%%%%%%%%%%%%%%%%%%%%%%

\section{Fisher Metrics}\label{sec3}

The generic state with maximally disordered subsystems parametrized by ${\boldsymbol{t}}$  reads
\begin{equation}
\label{rhot}
\rho_{\boldsymbol{t}}=\frac{1}{4}\sum_{n=1}^3t_{nn} \sigma_n\otimes\sigma_n, 
\end{equation}
and has a matrix representation (in the canonical basis):
\begin{equation}
\label{rhotmatrix}
\rho_{\boldsymbol t}=\left(
\begin{array}{cccc}
 \frac{t_{33}+1}{4} & 0 & 0 & \frac{t_{11}-t_{22}}{4} \\
 0 & \frac{1-t_{33}}{4} & \frac{t_{11}+t_{22}}{4} & 0 \\
 0 & \frac{t_{11}+t_{22}}{4} & \frac{1-t_{33}}{4} & 0 \\
 \frac{t_{11}-t_{22}}{4} & 0 & 0 & \frac{t_{33}+1}{4} \\
\end{array}
\right).
\end{equation}
We shall derive different Fisher metrics for the set of such states.

\subsection{Classical Fisher Metric in Phase Space}\label{subsec31}

The state \eqref{rhot} can also be represented in the phase space by means of the Husimi $Q$-function \cite{Husimi}, 
\begin{equation}
Q_{\boldsymbol t}\left( {\boldsymbol x}\right):=\frac{1}{4\pi^2}\langle \Omega_1| \langle \Omega_2| \rho_{\boldsymbol t} | \Omega_2\rangle |\Omega_1\rangle, 
\label{p}
\end{equation}
where 
\begin{eqnarray}
|\Omega_1\rangle&=&\cos {\frac{\theta_1}{2}}  | 0 \rangle+ e^{-i \phi_1} \sin{\frac{\theta_1}{2}}|1\rangle,
\\
|\Omega_2\rangle&=&\cos {\frac{\theta_2}{2}}  | 0 \rangle + e^{-i \phi_2} \sin{\frac{\theta_2}{2}}|1\rangle,
\end{eqnarray}
are $SU(2)$ coherent states \cite{KS}. Here, ${\boldsymbol x}:=(\theta_1,\theta_2,\phi_1,\phi_2)$ is the vector of random variables $\theta_{1,2}\in[0,\pi)$, $\phi_{1,2}\in[0,2\pi]${ ,} with probability measure 
$d{\boldsymbol x}=\frac{1}{16\pi^2}\sin\theta_1\sin\theta_2 d\theta_1d\theta_2 d\phi_1d\phi_2$.

Recall that the classical Fisher--Rao information metric for probability distribution functions $Q_{\boldsymbol t}$ is given by:
\begin{equation}
g^{FR}_{ij}:=\int Q_{\boldsymbol t}\left( {\boldsymbol x}\right) \partial_i \log  Q_{\boldsymbol t}\left( {\boldsymbol x}\right) \partial_j \log  Q_{\boldsymbol t}\left( {\boldsymbol x}\right) d {\boldsymbol x},
\label{gFisher}
\end{equation}
where $\partial_i :=\frac{\partial}{\partial t_{ii}}$. 

Considering \eqref{rhot} and \eqref{p}, we explicitly get
\begin{eqnarray}
Q_{\boldsymbol t}\left(\theta_1, \theta_2, \phi_1, \phi_2 \right)
&=&\frac{2\sin\theta_1 \sin \theta_2  }{64 \pi ^2} \Big[t_{11} \left(\cos(\phi_1+\phi_2)+\cos(\phi_1-\phi_2)\right) 
-t_{22} \left(\cos(\phi_1+\phi_2)-\cos(\phi_1-\phi_2)\right)  \Big] \notag \\
&+&\frac{4}{64 \pi ^2}\left[t_{33} \cos \theta_1 \cos \theta_2+1 \right].
\end{eqnarray}
Unfortunately, it is not possible to arrive at an analytical expression for  $g_{ij}$.
In fact, the integral \eqref{gFisher} can only be evaluated numerically.

Let us, however, analyze an important property of this metric.

\begin{prop}\label{prop1}
The metric \eqref{gFisher} is invariant under local rotations.
\end{prop}

\begin{proof}
From Section \ref{sec2}, we know that the equivalence relation $\rho\sim \left(U_1\otimes U_2\right)\rho
\left(U_1^\dag\otimes U_2^\dag\right)$
leads to the transformation \eqref{TOTO} for the parameters matrix,
where the correspondence between $O_i$s and $U_i$s, given by Equation \eqref{UnsU}, can be read as follows
\begin{equation}
\label{Correspondence}
(O_i)_{\mu\nu}=\frac{1}{2} {\rm Tr}\left[\sigma_\mu U_i\sigma_\nu U_i^\dagger\right].
\end{equation}
Therefore, the effect of the equivalence relation acts on $Q_{\boldsymbol t}\left( {\boldsymbol x}\right)$ of Equation \eqref{p} by changing the parameter vector ${\boldsymbol{t}}=(t_{11},t_{22},t_{33})$ via an orthogonal transformation given by \eqref{Correspondence}. That is
\begin{equation}
Q_{\boldsymbol t}( {\boldsymbol x})\mapsto Q_{{\boldsymbol t}^{\prime}}\left( {\boldsymbol x}\right),
\end{equation}
where ${\boldsymbol t}^{\prime}=O {\boldsymbol t}$, or explicitly
\begin{eqnarray}\label{transformation}
&&t^{\prime}_{11}=O_{11}t_{11}+O_{12}t_{22}+O_{13}t_{33}, \nonumber\\
&&t^{\prime}_{22}=O_{21}t_{11}+O_{22}t_{22}+O_{23}t_{33}, \\
&&t^{\prime}_{33}=O_{31}t_{11}+O_{32}t_{22}+O_{33}t_{33}\nonumber .
\end{eqnarray}
Let us now see how the entries of the Fisher--Rao metric \eqref{gFisher} change under an orthogonal transformation of $T$ given in terms of \eqref{Correspondence}. To this end, consider \eqref{gFisher} 
written for ${\boldsymbol t}^{\prime}$ as follows
\begin{eqnarray}\label{gFishervary}
\widetilde{g}^{FR}_{ij}&=&
\int Q_{{\boldsymbol t}^\prime}\left( {\boldsymbol x}\right) \partial^\prime_i \log  Q_{{\boldsymbol t}^\prime}\left( {\boldsymbol x}\right) \partial^\prime_j \log  Q_{{\boldsymbol t}^\prime}\left( {\boldsymbol x}\right) d {\boldsymbol x}\notag\\
&=&\int \frac{1}{Q_{ {\boldsymbol t}^{\prime}} \left( {\boldsymbol x}\right)} \partial^{\prime}_i Q_{ {\boldsymbol t}^{\prime}} \left( {\boldsymbol x}\right) \partial^{\prime}_j 
Q_{ {\boldsymbol t}^{\prime}} \left( {\boldsymbol x}\right) d{\boldsymbol x},
\end{eqnarray}
where $\partial^{\prime}_i:=\frac{\partial}{\partial t^{\prime}_{ii}}$. From \eqref{transformation}, we recognize the functional relation $t^{\prime}_{ii}=t^{\prime}_{ii}(\boldsymbol{t})$; therefore, we have
\begin{align}\label{transfderiv}
\partial_{i}Q_{ {\boldsymbol t}^{\prime}} \left( {\boldsymbol x}\right)
= \partial_1^{\prime} Q_{ {\boldsymbol t}^{\prime}} \left( {\boldsymbol x}\right)\partial_{i}t^{\prime}_{11}
+\partial_2^{\prime}Q_{ {\boldsymbol t}^{\prime}} \left( {\boldsymbol x}\right)\partial_{i}t^{\prime}_{22}
+\partial_3^{\prime}Q_{ {\boldsymbol t}^{\prime}} \left( {\boldsymbol x}\right)\partial_{i}t^{\prime}_{33}
\end{align}
Finally, from \eqref{gFishervary} and \eqref{transfderiv}, we arrive at
\begin{equation}
g^{FR}_{ij}=\sum_{m,n=1}^3 \widetilde{g}^{FR}_{mn}O_{i n}O_{j m},
\end{equation}
and as a consequence we immediately obtain $g^{FR}=O\ \widetilde{g}^{FR}\ O^\top$. 
\end{proof}

%%%%%%%%%%%%%%%%%%%%%%%%%%%%%%%%%%%%%%%%

\subsection{Quantum Fisher Metrics}\label{subsec32}

It is remarkable that \eqref{gFisher} is the unique (up to a constant factor) monotone Riemannian
metric (that is, the metric contracting under any stochastic map) in the class of probability 
distribution functions (classical or commutative case) \cite{Chentsov}.
This is not the case in an operator setting (quantum case), where the notion
of Fisher information has many natural generalizations due to non-commutativity \cite{Petz}.
Among  the  various  generalizations, two are distinguished. 

The first natural generalization of the classical Fisher information arises when
one formally generalizes the expression \eqref{gFisher}. This was, in fact, first done in a quantum estimation setting \cite{Helstrom}.
To see how this happens, note that in a symmetric form, the derivative of $Q_{\boldsymbol t}$ reads
\begin{equation}
\partial_i Q_{\boldsymbol t}=\frac{1}{2}\left(
\partial_i \log Q_{\boldsymbol t} \cdot Q_{\boldsymbol t}
+Q_{\boldsymbol t} \cdot \partial_i \log Q_{\boldsymbol t}
\right).
\end{equation}
In Equation \eqref{gFisher}, replacing the integration by trace,
$Q_{\boldsymbol t}$  by $\rho_{\boldsymbol t}$, and the logarithmic derivative 
$\partial_i \log Q_{\boldsymbol t}$ 
by the symmetric logarithmic derivative $L_i$ determined by
\begin{equation}
\label{eqSLD}
\partial_i \rho_{\boldsymbol t}=\frac{1}{2}\left(L_i\rho_{\boldsymbol t}+\rho_{\boldsymbol t} L_i\right),
\end{equation}
we come to the quantum Fisher information (derived via the symmetric logarithmic
derivative)
\begin{equation}
\label{gH}
g^{H}_{ij}(\rho_\theta):=\frac{1}{2}{\rm Tr}\left(L_i L_j \rho_{\boldsymbol t}+L_jL_i\rho_{\boldsymbol t}\right).
\end{equation}
Using \eqref{rhotmatrix} to solve \eqref{eqSLD} yields
\begin{eqnarray}
\label{L1}
L_{1}&=&\left(
\begin{array}{cccc}
 \frac{t_{11}-t_{22}}{\left ( t_{11}-t_{22} \right)^2-(t_{33}+1)^2} & 0 & 0 & \frac{t_{33}+1}{-\left ( t_{11}-t_{22} \right)^2+(t_{33}+1)^2} \\
 0 & \frac{t_{11}+t_{22}}{\left ( t_{11}+t_{22} \right)^2-(t_{33}-1)^2} & \frac{t_{33}-1}{\left ( t_{11}+t_{22} \right)^2-(t_{33}-1)^2} & 0 \\
 0 & \frac{t_{33}-1}{\left ( t_{11}+t_{22} \right)^2-(t_{33}-1)^2} & \frac{t_{11}+t_{22}}{\left ( t_{11}+t_{22} \right)^2-(t_{33}-1)^2} & 0 \\
 \frac{t_{33}+1}{-\left ( t_{11}-t_{22} \right)^2+(t_{33}+1)^2} & 0 & 0 & \frac{t_{11}-t_{22}}{\left ( t_{11}-t_{22} \right)^2-(t_{33}+1)^2} \\
\end{array}\right), 
\\ \nonumber \\
\label{L2}
L_{2}&=& \left(
\begin{array}{cccc}
 \frac{t_{22}-t_{11}}{\left ( t_{11}-t_{22} \right)^2-(t_{33}+1)^2} & 0 & 0 & \frac{t_{33}+1}{\left ( t_{11}-t_{22} \right)^2-(t_{33}+1)^2} \\
 0 & \frac{t_{11}+t_{22}}{\left ( t_{11}+t_{22} \right)^2-(t_{33}-1)^2} & \frac{t_{33}-1}{\left ( t_{11}+t_{22} \right)^2-(t_{33}-1)^2} & 0 \\
 0 & \frac{t_{33}-1}{\left ( t_{11}+t_{22} \right)^2-(t_{33}-1)^2} & \frac{t_{11}+t_{22}}{\left ( t_{11}+t_{22} \right)^2-(t_{33}-1)^2} & 0 \\
 \frac{t_{33}+1}{\left ( t_{11}-t_{22} \right)^2-(t_{33}+1)^2} & 0 & 0 & \frac{t_{22}-t_{11}}{\left ( t_{11}-t_{22} \right)^2-(t_{33}+1)^2} \\
\end{array}
\right),
\\ \nonumber \\
\label{L3}
L_{3}&=& \left(
\begin{array}{cccc}
 \frac{t_{33}+1}{-\left ( t_{11}-t_{22} \right)^2+(t_{33}+1)^2} & 0 & 0 & \frac{t_{11}-t_{22}}{\left ( t_{11}-t_{22} \right)^2-(t_{33}+1)^2} \\
 0 & \frac{1-t_{33}}{\left ( t_{11}+t_{22} \right)^2-(t_{33}-1)^2} & -\frac{t_{11}+t_{22}}{\left ( t_{11}+t_{22} \right)^2-(t_{33}-1)^2} & 0 \\
 0 & -\frac{t_{11}+t_{22}}{\left ( t_{11}+t_{22} \right)^2-(t_{33}-1)^2} & \frac{1-t_{33}}{\left ( t_{11}+t_{22} \right)^2-(t_{33}-1)^2} & 0 \\
 \frac{t_{11}-t_{22}}{\left ( t_{11}-t_{22} \right)^2-(t_{33}+1)^2} & 0 & 0 & \frac{t_{33}+1}{-\left ( t_{11}-t_{22} \right)^2+(t_{33}+1)^2} \\
\end{array}
\right){\blu .} \;\;
\end{eqnarray}
Then, from Equation \eqref{gH}, we obtain
{\small
\begin{equation}
\label{gHmatrix}
g^{H}=\frac{1}{\Delta}\left(
\begin{array}{ccc}
 1- \|{\boldsymbol t}\|^2 -2 t_{11} t_{22}  t_{33} 
 & (1+ \| {\boldsymbol t}\|^2 -2 t_{33}^2) t_{33}+2  t_{11}  t_{22} 
 & \left(1+ \| {\boldsymbol t} \|^2- 2t_{22}^2\right)  t_{22}+2  t_{11}  t_{33} \\ \\
(1+ \| {\boldsymbol t}\|^2 -2 t_{33}^2) t_{33}+2  t_{11}  t_{22}  & 
 1- \|{\boldsymbol t}\|^2 -2 t_{11} t_{22}  t_{33}
 &\left(1+ \| {\boldsymbol t} \|^2-2 t_{11}^2\right)  t_{11}+2  t_{22}  t_{33} \\ \\
 \left(1+ \| {\boldsymbol t} \|^2- 2t_{22}^2\right)  t_{22}+2  t_{11}  t_{33}
&\left(1+ \| {\boldsymbol t} \|^2-2 t_{11}^2\right)  t_{11}+2  t_{22}  t_{33}
&  1- \|{\boldsymbol t}\|^2 -2 t_{11} t_{22}  t_{33}  \\
\end{array}
\right),
\end{equation}
}
where 
\begin{equation}
\Delta:=\left(\left ( t_{11}+t_{22} \right)^2-(1- t_{33})^2\right) \left(\left ( t_{11}-t_{22} \right)^2-(1+ t_{33})^2\right).
\end{equation}
It follows 
\begin{equation}
\det g^{H}=\frac{1}{\Delta}.
\end{equation}

\bigskip

The second generalization of the classical Fisher information arises 
by noticing that \eqref{gFisher} can be equivalently expressed as 
\begin{equation}
g^{FR}_{ij}=4\int  \partial_i \sqrt{Q_{\boldsymbol t}\left( {\boldsymbol x}\right)}
\partial_j \sqrt{ Q_{\boldsymbol t}\left( {\boldsymbol x}\right)} \, d {\boldsymbol x}.
\end{equation}
Replacing the integration by trace, and the parameterized probabilities $Q$
by a parameterized density operator $\rho$, we can arrive at \cite{Luo}
\begin{equation}\label{gWY}
g^{WY}_{ij}:=4 {\rm Tr} \left[ \left(\partial_i  \sqrt{\rho_{\boldsymbol t}}\right) 
\left(\partial_j  \sqrt{\rho_{\boldsymbol t}}\right)\right].
\end{equation} 
The superscript indicates the names of Wigner and Yanase since Equation \eqref{gWY} is motivated by their work on skew information \cite{WigYan}.

The following proposition relates $g^{WY}$ and $g^{H}$.

\begin{prop}\label{prop2}
If $\left[\rho_{\boldsymbol t},L_i\right]=0$, $\forall i$, then $g^{WY}=g^{H}$.
\end{prop}

\begin{proof}
If $\left[\rho_{\boldsymbol t},L_i\right]=0$, $i=1,2,3$ from \eqref{eqSLD} it follows
$\partial_i \rho_{\boldsymbol t}=\rho_{\boldsymbol t} L_i$ and $g_{ij}^{H}={\rm Tr} 
\left[\rho_{\boldsymbol t} L_iL_j\right]$. It is also 
\begin{eqnarray}
g_{ij}^{ WY}&=&{\rm Tr} \left[ {\rho_{\boldsymbol t}}^{-1/2} \left(\partial _i {\rho_{\boldsymbol t}}\right)
{\rho_{\boldsymbol t}}^{-1/2}\left(\partial _j {\rho_{\boldsymbol t}}\right)
\right]\label{line1}\\
&=&{\rm Tr} \left[ {\rho_{\boldsymbol t}}^{-1/2}  {\rho_{\boldsymbol t}} L_i
{\rho_{\boldsymbol t}}^{-1/2} {\rho_{\boldsymbol t}} L_j
\right]\label{line2}\\
&=&{\rm Tr} \left[ {\rho_{\boldsymbol t}}^{1/2}  L_i
{\rho_{\boldsymbol t}}^{1/2} L_j
\right]\label{line3}\\
&=&{\rm Tr} \left[ {\rho_{\boldsymbol t}}  L_i  L_j
\right]=g_{ij}^{H},
\label{line4}
\end{eqnarray}
where in \eqref{line1} we have used the definition \eqref{gWY}; in \eqref{line2} we have used the result 
$\partial_i \rho_{\boldsymbol t}=\rho_{\boldsymbol t} L_i$; and in \eqref{line4} the commutativity of $\rho_{\boldsymbol t}$ and $L_i$s.
\end{proof}

\medskip

By referring to \eqref{rhotmatrix} and \eqref{L1}--\eqref{L3}, we can see that 
the conditions of Proposition \ref{prop2} are satisfied, hence hereafter we shall consider
$g^{WY}=g^{H}$.

%%%%%%%%%%%%%%%%%%%%%%%%%%%%%%%%%%%%%%%%%%%%%%%%%%%%%%

\section{Volume of States with Maximally Disordered Subsystems}\label{sec4}

The volume of two-qubit states with maximally disordered subsystems will be given by
\begin{equation}
\label{V}
V=\int_{\cal T} \sqrt{ \det g} \; d{\boldsymbol t},
\end{equation}
where the tetrahedron $\cal T$ is characterized by  equations:
\begin{eqnarray}
1-t_{11}-t_{22}-t_{33} &\geq& 0, \notag \\
1-t_{11}+t_{22}+t_{33} &\geq& 0,\notag \\
1+t_{11}-t_{22}+t_{33} &\geq& 0,\notag \\
1+t_{11}+t_{22}-t_{33} &\geq& 0. 
\end{eqnarray}
{(If the Fisher metric is degenerate, then the volume is meaningless, given that $\det g = 0$. This reflects on the parameters describing the states. Indeed, at least one of them should be considered as depending on the others. Thus, one should restrict the attention to a proper submanifold.)}

Analogously, the volume of two-qubit separable states with maximally disordered subsystems
will be given by
\begin{equation}
\label{Vs}
V_s=\int_{\cal O} \sqrt{ \det g} \; d{\boldsymbol t},
\end{equation}
where the octahedron $\cal O$ is characterized by  equations:
\begin{eqnarray}
1-t_{11}-t_{22}-t_{33} &\geq& 0, \notag\\
1+t_{11}-t_{22}-t_{33} &\geq& 0,\notag \\
1+t_{11}+t_{22}-t_{33} &\geq& 0,\notag \\
1-t_{11}+t_{22}-t_{33} &\geq& 0, \notag\\
1-t_{11}-t_{22}+t_{33} &\geq& 0, \notag \\
1+t_{11}-t_{22}+t_{33} &\geq& 0,\notag \\
1+t_{11}+t_{22}+t_{33} &\geq& 0,\notag \\
1-t_{11}+t_{22}+t_{33} &\geq& 0. 
\end{eqnarray}

Concerning the classical Fisher metric,
as a consequence of the result $g^{FR}=O\ \widetilde{g}^{FR} O^\top$ (Proposition \ref{prop1}), we have that the volume computed as $\int \sqrt{\det g^{FR}}\ d{\boldsymbol{t}}$ is invariant under orthogonal transformations of the parameters matrix $T$.

However, the integral \eqref{gFisher} can only be performed numerically.
To this end, we have generated $10^3$ points randomly distributed inside the tetrahedron 
${\cal T}$. On each of these points, the integral \eqref{gFisher} has been numerically evaluated, hence the value of the function $\sqrt{\det g^{FR}}$ determined on a set of discrete points.
After that, data has been interpolated and a smooth function $\sqrt{\det g^{FR}}(t_{11},t_{22},t_{33})$ obtained.
This has been used to compute the quantities \eqref{V}, \eqref{Vs} and their ratio, resulting as 
$V=0.168$, $V_s=0.055$, $V_s/V=0.327$.

%%%%%%%%%%%%%%%%%%%%%%%%%%%%%%%%%%%%%%%%%

The same volumes computed by the quantum Fisher metric 
result in $V=\pi^2$, $V_s=(4-\pi)\pi$ and their ratio is $V_s/V=(4-\pi)/\pi\approx 0.27$. 

%%%%%%%%%%%%%%%%%%%%%%%%%%%%%%%%%%%%%%%%%

\bigskip

It would also be instructive to see how the ratio of the volume of separable states and the volume of total states varies versus the purity.
This latter quantity turns out to be
\begin{equation}
P= {\rm Tr} \left(\rho^2\right)=\frac{1}{4}\left( 1+ \|{\boldsymbol t}\|^2 \right)
\end{equation}
Thus, fixing a value of $P\in[\frac{1}{4},1]$ amounts to fixing a sphere $\cal S$ centered in the origin of 
$\mathbb{R}^3$ and with a radius $\sqrt{4P-1}$.
Hence, the volume of states with fixed purity will be given by  
\begin{equation}
V(P)=\int_{{\cal T}\cap{\cal S}} \sqrt{ \det g} \; d{\boldsymbol t},
\label{Vpurity}
\end{equation}
while the volume of separable states with fixed purity will be given by  
\begin{equation}
V_s(P)=\int_{{\cal O}\cap{\cal S}} \sqrt{ \det g} \; d{\boldsymbol t}.
\label{Vspurity}
\end{equation}
As a consequence, we can obtain the ratio
\begin{equation}
R(P)=\frac{V_{s}(P)}{V(P)},
\label{RP}
\end{equation}
as a function of purity $P$.
Such a ratio is plotted in Figure \ref{fig1} as a solid (resp. dashed) line for the classical (resp. quantum) Fisher metric. There, it is shown that by increasing the purity, the volume of separable states diminishes until it becomes a null measure. In fact, this happens already at purity $P=1/2$.
However, when the purity is low enough (below $1/3$), all states are separable ($R=1$). 
Above all, the two curves show the same qualitative behavior.
This means that the usage of classical Fisher information on the phase space representation of quantum states (as a true probability distribution function) is able to capture the main feature of the geometry of quantum states.

\begin{figure}[h]
\centering
\includegraphics[width=8cm]{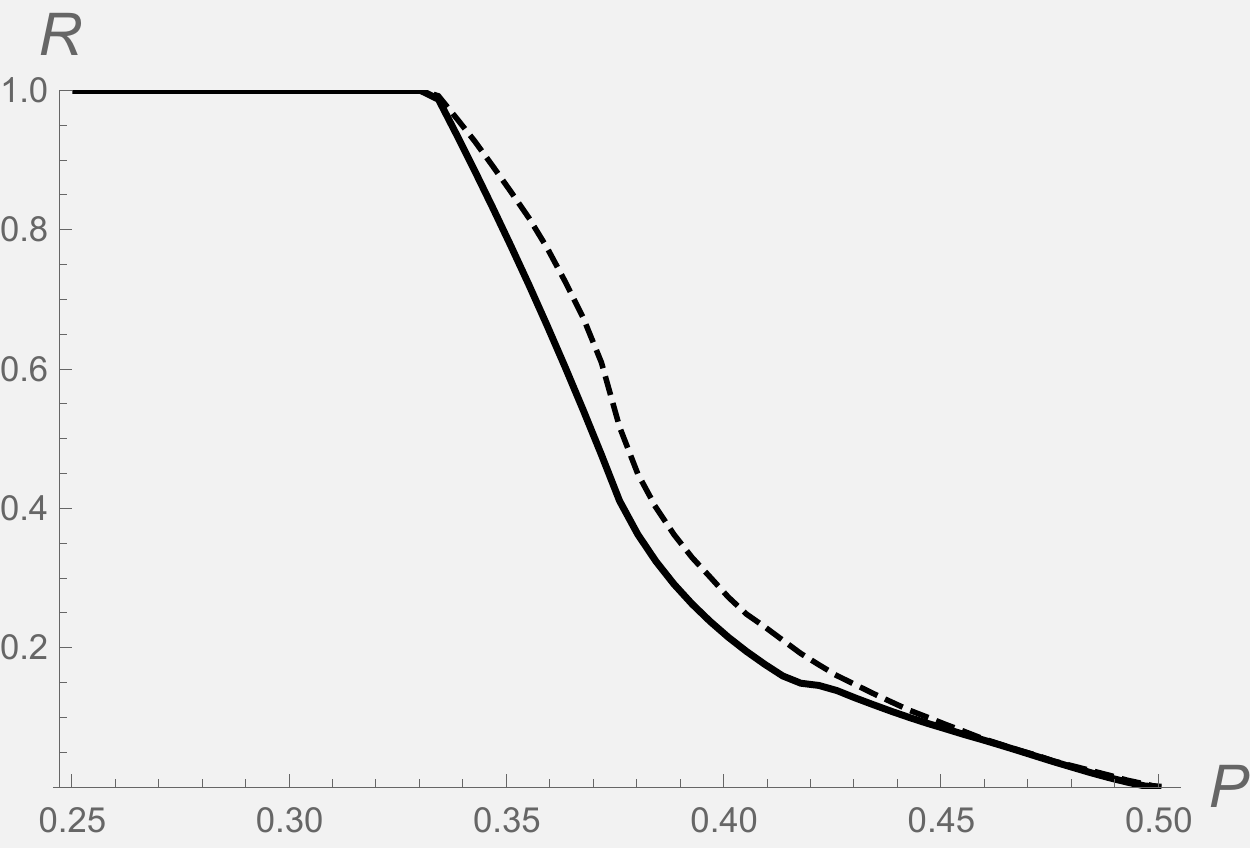}
\caption{Ratio of volumes \eqref{RP} vs purity $P$. Solid line refers to the classical Fisher metric; dashed line refers to the quantum Fisher metric.}
\label{fig1}
\end{figure}

%%%%%%%%%%%%%%%%%%%%%%%%%%%%%%%%%%%%%%%%%%%%%%%%%%%%%%

\section{Conclusion}\label{sec5}

In conclusion, we have investigated the volume of the set of two-qubit states with maximally disordered subsystems by considering their phase space representation in terms of probability distribution functions and by applying the classical Fisher information metric to them.
The results have been contrasted with those obtained by using quantum versions of the Fisher metric.
Although the absolute values of volumes of separable and entangled states turn out to be different in the two approaches, their ratios are comparable. 
Above all, the behavior of the volume of sub-manifolds of separable and entangled states with fixed purity is shown to be almost the same in the two approaches.

Thus, we can conclude that classical Fisher information in phase space is able to capture the features of the volume of quantum states. The question then arise as to which aspects such an approach will be unable to single out with respect to the purely quantum one. Besides that, our work points out other interesting issues
within the quantum metrics. In fact, in the considered class of two-qubit states, the Helstrom and the 
Wigner--Yanase-like quantum Fisher metrics coincide. This leads us to ask the following questions: For which more general class of states in finite dimensional systems are the two equal? When they are not equal, 
what is their order relation? These investigations are left for future works.

Additionally, it is worth noting that 
our approach to the volume of states by information geometry offers the possibility to characterize quantum logical gates. 
In fact, given a logical gate (unitary) $G$ on two-qubit states, 
the standard entangling power is defined as \cite{ZZF00}:
$$
\mathfrak{E}(G):=\int E\left( G |\psi \rangle \right)
d\mu\left( |\psi\rangle\right),
$$
where $E$ is an entanglement quantifier \cite{Horo} and the overall average is the
product states $|\psi\rangle=|\psi_1\rangle \otimes |\psi_2\rangle$
according to a suitable measure $\mu\left( |\psi\rangle\right)$.
It is then quite natural to take the measure induced by the Haar measure on 
$SU(2)\otimes  SU(2)$ \cite{NTM16}.
However, $\mathfrak{E}(G)$ is not usable on the subset of states with maximally disordered subsystems as it is restricted to pure states.
Here, our measure comes into play which leads to the following: 
$$
\mathfrak{E}(G):=\int_{\cal O} E\left( G \rho_{\boldsymbol t} \, G^\dag \right) \sqrt{\det g} \,
d\boldsymbol{t}.
$$
In turn, this paves the way to a general formulation that involves the average overall separable states by also including parameters $\boldsymbol{r},\boldsymbol{s}$. Clearly, this would provide the most accurate characterization of the entangling potentialities of $G$. 

Finally, our approach can be scaled up to three or more qubit, but since analytical calculations will soon become involved, one should consider families of states with a low number of parameters, e.g., those proposed in \cite{Altafini}. Nevertheless, such families can provide geometrical insights for more general cases.

%%%%%%%%%%%%%%%%%%%%%%%%%%%%%%%%%%%%%%%%%%%%%%%%%%%%%%

\acknowledgments{The work of M.R. is supported by China Scholarship Council.}

%%%%%%%%%%%%%%%%%%%%%%%%%%%%%%%%%%%%%%%%%%%%%%%%%%%%%%
%%%%%%%%%%%%%%%%%%%%%%%%%%%%%%%%%%%%%%%%%%

\end{document}